\documentclass[letterpaper, 10pt, conference]{ieeeconf}      
\IEEEoverridecommandlockouts                             
\usepackage{graphics} 
\usepackage{epsfig} 
\usepackage{mathptmx} 
\usepackage{times} 
\usepackage{amsmath} 
\usepackage{amssymb}
\usepackage{setspace}
\renewcommand{\baselinestretch}{0.93}

\usepackage[colorinlistoftodos, textwidth=3cm, shadow,disable]{todonotes}
\usepackage{multirow}
\usepackage{dirtytalk}

\newtheorem{assumption}{Assumption}
\newtheorem{proposition}{Proposition}
\newtheorem{theorem}{Theorem}
\newtheorem{lemma}{Lemma}

\title{\LARGE \bf
Detection of Cyber-Attacks in Collaborative Intersection Control
}

\author{Twan Keijzer$^{1}$, Fabian Jarmolowitz$^{2}$, and Riccardo M.G. Ferrari$^{1}$
\thanks{$^{1}$Twan Keijzer and Riccardo M.G. Ferrari are with Delft Centre for Systems and Control,
        Delft University of Technology, 2628 CD Delft, The Netherlands
        {\tt\small \{t.keijzer,r.ferrari\}@tudelft.nl}}%
\thanks{$^{2}$Fabian Jarmolowitz is with Corporate Sector Research and Advanced Engineering, Robert-Bosch GmbH, 
				71272 Renningen, Germany
{\tt\small fabian.jarmolowitz@bosch.com}}%
}
\begin{document}

\maketitle
\thispagestyle{empty}
\pagestyle{empty}

\begin{abstract} 
Road intersections are widely recognized as a lead cause for accidents and traffic delays. In a future scenario with a significant adoption of Cooperative Autonomous Vehicles, solutions based on fully automatic, signage-less Intersection Control would become viable. Such a solution, however, requires communication between vehicles and, possibly, the infrastructure over wireless networks. This increases the attack surface available to a malicious actor, which could lead to dangerous situations. In this paper, we address the safety of Intersection Control algorithms, and design a Sliding-Mode-Observer based solution capable of detecting and estimating false data injection attacks affecting vehicles' communication. With respect to previous literature, a novel detection logic with improved detection performances is presented. Simulation results are provided to show the effectiveness of the proposed approach.
\end{abstract}

\section{Introduction}
\label{sec:introduction}
\noindent Classical road intersections for human-driven vehicles are managed via fixed signage and traffic lights, which allows for a sub-optimal vehicle throughput while guaranteeing an adequate safety level. Still, it is well known that classical intersections are a major cause of accidents, due to human error, and traffic delays \cite{european2016traffic}. In the early 2000s, based on the projected introduction of \emph{Cooperative Autonomous Vehicles} (CAV), works such as \cite{dresner2005multiagent} proposed replacing classical intersections with safer, automated solutions. \emph{Intersection Control} (IC) would thus automatize the tasks of negotiating, planning and executing the trajectories of CAVs in order to increase the safety and vehicle throughput of the junction.

A key enabling technology of an IC solution are wireless networks, allowing real-time communication of measurements and control signals between CAVs, and to the road-side infrastructure. The security and robustness of such Vehicle-to-Everything (V2X) networks, is thus of paramount importance for the safety of the IC itself. As a first line of defence, current V2X protocols include encryption and authentication mechanisms to prevent intrusion. For instance, the Autosar standard with End-to-End Protection (E2E), following the ISO 26262 standard, is such a preventive protection measure and is analyzed in \cite{bondavalli_making_2014}.
The case of communication disruption or false data injection by a malicious attacker \cite{amoozadeh2015security,yan2016can} that can circumvent these protections would instead require a different approach. Indeed, a major difference between an ordinary fault and a smart attacker is that the latter actively seeks to cause great harm to the IC, while minimizing the possibility of being detected \cite{engoulou2014vanet,van2017analyzing}. This lead to development of increasingly sophisticated anomaly detection approaches for V2X communication traffic, as a second line of defence that can provide a guarantee against either inside attackers, or sophisticated attackers that successfully infiltrated the system. These techniques include plausibility checks based on elementary models of the CAVs and IC \cite{bissmeyer2012assessment,dietzel2012graph,yan2008providing}, as well as more advanced, model-based approaches that have been proposed in the literature for detecting attacks in general Cyber-Physical Systems (CPS), such as \cite{pasqualetti2013attack,quan2018distributed,Jahanshahi_2018aa,Keijzer2019,Keijzer_2021aa}. Nevertheless, these techniques were not yet applied to CAVs on an automated intersection.

In this paper, we will address the problem of designing a second line of defence for IC subjected to false-data injection attacks. To this end, a novel cyber-attack detection method is presented based on a \emph{Sliding Mode Observer} (SMO). In previous work by the authors \cite{Keijzer2019} a detection logic based on the so-called \emph{Equivalent Output Injection} (EOI) term of the SMO was presented. The novel approach proposed in the present paper no longer requires this EOI, leading to better detection performance. The proposed technique is verified in simulation of an IC scenario. Here a control approach called \emph{Virtual Platooning} (VP,  \cite{medina2017cooperative}) is used. However, also potentially better-performing  optimization-based approaches could be applied \cite{campos2014cooperative,kamal2014vehicle,katriniok2017distributed,kneissl2018feasible} without affecting the proposed detection method.

The structure of the paper is as follows: Section \ref{sec:problem_statement} introduces the problem addressed in this paper, while Sections \ref{sec:smo} and \ref{sec:detection}, respectively, introduce the SMO design and the detection thresholds on which the proposed attack detection strategy is built. Simulation results are shown in Section \ref{sec:sim}, while concluding remarks are finally drawn in Section \ref{sec:conclusion}.
\section{Problem Formulation}
\label{sec:problem_statement}
\noindent In this paper, detection of cyber-attacks on the inter-vehicle communication is considered in a collaborative intersection control scenario. Each car is modeled as
\begin{equation}\label{eq:car_model}
    \left\{
\begin{aligned}
\left[\begin{matrix}\dot{p}_i\\\dot{v}_i\\\dot{a}_i\end{matrix}\right]&=\underset{A_i}{\underbrace{\left[\begin{matrix}0&1&0\\0&0&1\\0&0&-\frac{1}{\tau_i}\end{matrix}\right]}}\underset{x_i}{\underbrace{\left[\begin{matrix}p_i\\v_i\\a_i\end{matrix}\right]}}+\underset{B_i}{\underbrace{\left[\begin{matrix}0\\0\\\frac{1}{\tau_i}\end{matrix}\right]}}u_i\,,\\
y_i &= \left[\begin{matrix}p_i-p_{i-1}-L_{i}\\v_i-v_{i-1}\\v_i\\a_i\end{matrix}\right]+\zeta_i\,,
 \end{aligned}\right.
\end{equation}
where the subscripts $i$ and $i-1$ denote the variables are related to cars $i$ and $i-1$ respectively. $p$, $v$, $a$, $u$, $y$, $\zeta$, $\tau$, and $L$ are, respectively, the distance from the rear of the vehicle to the intersection (negative when approaching the intersection), velocity, acceleration, input, measurements, sensor noise, engine time constant, and length of the cars.

The considered IC scenario is depicted in Figure \ref{fig:IC_scenario}. In this scenario, as can be seen from Equation \eqref{eq:car_model}, each car measures the relative distance to the intersection, and relative velocity from the preceding car, as well as its own velocity and acceleration. Mandatory for having these measurements for all cars entering the intersection are either a central infrastructure with appropriate sensors and V2X communication and/or cars equipped with lateral sensors, e.g. \cite{Bosch_sensor2021}. 
Furthermore, each car receives, via wireless V2V communication, the input of the preceding car.

The interaction between two cars in an IC scenario can be modeled, from the perspective of car $i$, as
\begin{equation}\label{eq:IC_model}
    \left\{
\begin{aligned}
\left[\begin{matrix}\dot{x}_{i-1}\\\dot{x}_i\end{matrix}\right] &= \underset{A}{\underbrace{\left[\begin{matrix}A_{i-1}&0\\0&A_i\end{matrix}\right]}}\underset{x}{\underbrace{\left[\begin{matrix}x_{i-1}\\x_i\end{matrix}\right]}}+\underset{B}{\underbrace{\left[\begin{matrix}B_{i-1}&0\\0&B_i\end{matrix}\right]}}\underset{u}{\underbrace{\left[\begin{matrix}u_{i-1}\\u_i\end{matrix}\right]}}\,,\\
 y_i &= C\underset{x}{\underbrace{\left[\begin{matrix}x_{i-1}\\x_i\end{matrix}\right]}}+\underset{c}{\underbrace{\left[\begin{matrix}-L_{i}\\0_{3\times 1}\end{matrix}\right]}}+\zeta_i\,.
 \end{aligned}\right.
\end{equation}
Here, $C$ can be derived from equation \eqref{eq:car_model}. Furthermore, one can see that the coupling between the vehicles appears in the measurement equation only. These measurements, as well as the communicated input $u_{i-1}$, can be used by car $i$ to calculate a control input $u_i$ such that the IC objective is achieved. In this work, which primarily deals with the cyber-attack detection, any control law for IC can be chosen without affecting the detection method.
\begin{figure}[h!]
	\centering
	\includegraphics[trim=0 0 0 0mm,clip, height=0.5\columnwidth]{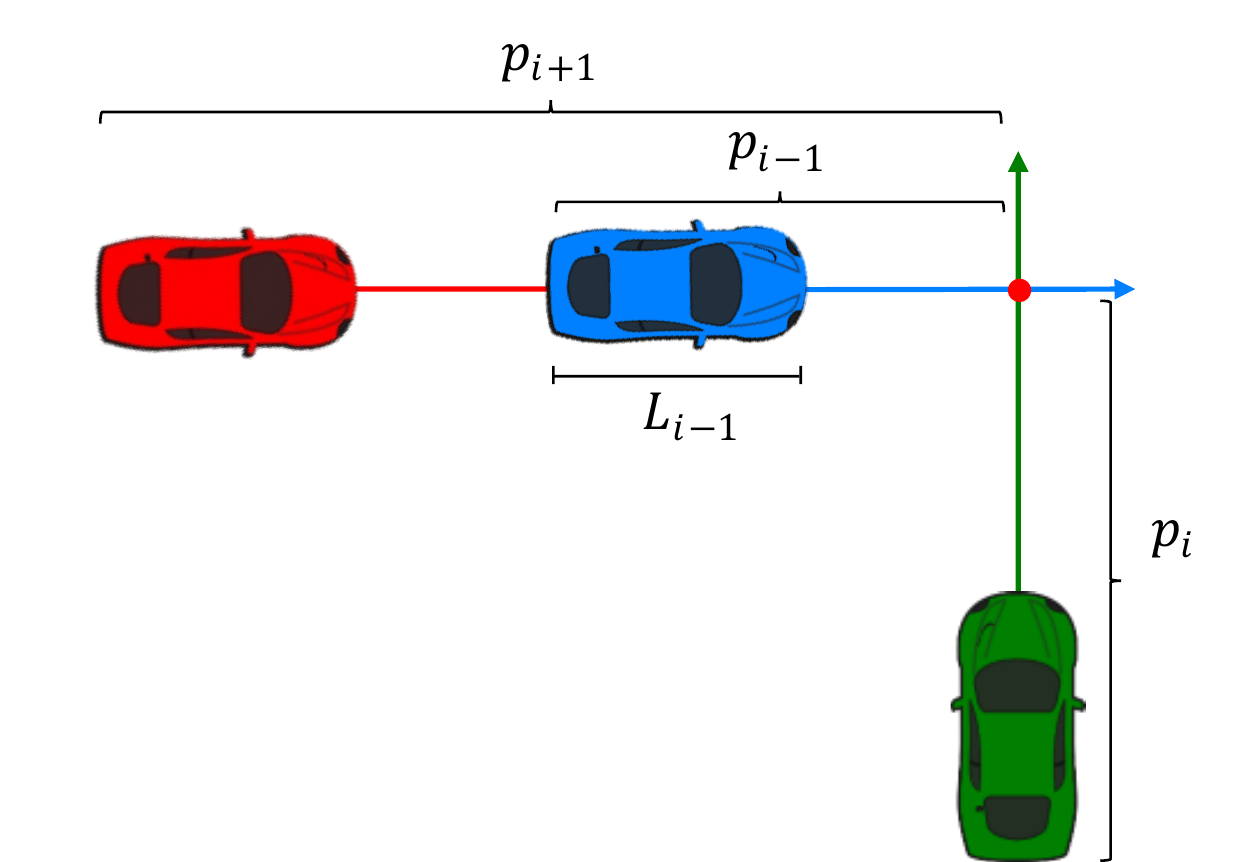}
	\caption{Intersection Control Scenario \cite{ElineThesis}}
	\label{fig:IC_scenario}
\end{figure}
\subsection{Model Uncertainty \& Cyber-attack}
\noindent In this section, System \eqref{eq:IC_model} will be rewritten to make the model uncertainty and cyber-attack explicit. To quantify the model uncertainty, the following assumption is made.

\begin{assumption}
Each car $i$ is assumed know its own dynamics, represented by $\tau_i$ and $L_i$.  It will, however, only have access to the nominal dynamics of the preceding car.
\end{assumption}

The nominal value of $\tau_{i-1}$, available to car $i$, is defined as
\begin{equation*}
    \hat{\tau}_{i-1} = r_\tau \tau_{i-1}\,.
\end{equation*}
Furthermore, the cyber-attack is defined as a \say{man in the middle} attack, such that each car $i$ will receive
\begin{equation*}\label{eq:attack}
    u_{i-1,r} = u_{i-1} + \Delta u_{i-1}\,,
\end{equation*}
where $\Delta u_{i-1}$ is the cyber-attack. To make the model uncertainty and cyber-attack explicit a model of the form
\begin{equation}\label{eq:IC_model_attacked}
        \left\{
\begin{aligned}
\dot{x}&=Ax+Bu+E\eta+F\Delta u_{i-1}\,,\\
 y_i &= Cx+c+\zeta_i\,,
 \end{aligned}\right.
\end{equation}
is proposed. Here $A$, $B$, and $u$ are redefined using
\begin{equation*}
    \tau_{i-1}\leftarrow\hat{\tau}_{i-1} \mbox{ and } u_{i-1} \leftarrow u_{i-1,r}\,,
\end{equation*}
such that they are known to car $i$, i.e. they can be used in its detection logic. Furthermore, $E$, $F$, and $\eta$ are defined as
\begin{equation*}
    E=\left[\begin{matrix}
0_{2\times 1}\\\frac{1}{\hat{\tau}_{i-1}}\\0_{3\times 1}
\end{matrix}\right]\,; F=\left[\begin{matrix}
0_{2\times 1}\\
-\frac{1}{\hat{\tau}_{i-1}}\\
0_{3\times 1}
\end{matrix}\right]\,; \eta=(r_\tau -1)(u_{i-1}-a_{i-1})\,,
\end{equation*}
where $E$ and $F$ are known to car $i$ and $\eta$ is unknown uncertainty. Note that Systems \eqref{eq:IC_model_attacked} and \eqref{eq:IC_model} are only reformulated.

The following assumptions are made on system \eqref{eq:IC_model_attacked}.
\begin{assumption}\label{ass:noise}
	The sensor noise $\zeta_i$ is zero-mean and bounded by a known value $\bar{\zeta}_i\geq|\zeta_i|$.
\end{assumption}
\begin{assumption}\label{ass:uncertainty}
	The uncertainty $\eta$, and cyber-attack $\Delta u_{i-1}$ are bounded by known values $\bar{\eta}\geq|\eta|$ and $\bar{\Delta}\geq|\Delta u_{i-1}|$.
\end{assumption}
These bounds are defined for the IC scenario in section \ref{sec:sim}.

\subsection{Model Transformation}
\noindent In System \eqref{eq:IC_model_attacked} the cyber-attack appears as an unknown input. This allows for the use of an SMO for estimation and detection of the cyber-attack \cite{Keijzer2019}. In order to implement this SMO based approach, the system is transformed to
\begin{equation}\label{eq:SMO_general_model}
    \left\{\begin{aligned}
	\dot{x}_1 &= A_{11} x_1 + A_{12} x_2 +B_1u + E_1\eta\ + F_1 \Delta u_{i-1}\\
	\dot{x}_2 &= A_{21} x_1 + A_{22} x_2 +B_2u + E_2\eta\ + F_2 \Delta u_{i-1}\\
	y_i&= x_2 + c + \zeta_i
	\end{aligned}\right.
\end{equation}
using a transformation introduced by \cite{Edwards2000-vs}. Here it is required that $A_{11}\prec0$ to ensure that the observer error dynamics are stable. The following assumption ensures $A_{11}\prec0$ \cite{Edwards2000-vs}.
\begin{assumption}\label{ass:zeros}
The invariant zeros of (A,F,C) lie in $\mathbb{C}_-$. 
\end{assumption}

\section{Sliding Mode Observer Design}
\label{sec:smo} 
\noindent Based on \cite{Keijzer2019}, the following SMO is introduced
\begin{equation}\label{eq:SMO}
    \left\{\begin{aligned}
\left[\begin{matrix}
\dot{\hat{x}}_1\\\dot{\hat{x}}_2
\end{matrix}\right]&=
\left[\begin{matrix}
A_{11}&A_{12}\\A_{21}&A_{22}
\end{matrix}\right]
\left[\begin{matrix}
\hat{x}_1\\\hat{x}_2
\end{matrix}\right]+
\left[\begin{matrix}
B_1\\B_2
\end{matrix}\right]
u-
\left[\begin{matrix}
A_{12}\\A_{22}^{-s}
\end{matrix}\right]
e_y+
\left[\begin{matrix}
0\\\nu
\end{matrix}\right]\\
\hat{y}_i&=\hat{x}_2+c\\
\nu&\triangleq-M \mbox{sgn}(e_y)
\end{aligned}\right.
\end{equation}
where $A_{22}^{-s}=A_{22}-A_{22}^s$, $A_{22}^s\prec0$ is the linear gain, the diagonal matrix $M\succ0$ is the switching gain, and $e_y\triangleq \hat{y}_i-y_i$.

Based on the system dynamics of equation \eqref{eq:SMO_general_model} and the SMO of equation \eqref{eq:SMO}, the observer error dynamics become
\begin{equation}\label{eq:SMO_err_dyn}
\resizebox{0.91\hsize}{!}{$
\left[\begin{matrix}
\dot{e}_1\\\dot{e}_2
\end{matrix}\right]=
\left[\begin{matrix}
A_{11}&0\\A_{21}&A_{22}^s
\end{matrix}\right]
\left[\begin{matrix}
e_1\\e_2
\end{matrix}\right]+
\left[\begin{matrix}
A_{12}\\A_{22}^{-s}
\end{matrix}\right]
\zeta_i
-\left[\begin{matrix}
E_1\\E_2
\end{matrix}\right]\eta-
\left[\begin{matrix}
F_1\\F_2
\end{matrix}\right]\Delta u_{i-1}+
\left[\begin{matrix}
0\\\nu
\end{matrix}\right].$}
\end{equation}
where $e_1 \triangleq \hat{x}_1-x_1$ and $e_2 \triangleq \hat{x}_2-x_2$. Furthermore, $e_y$ can be expressed in terms of $e_2$ as $e_y=e_2-\zeta_i$.

Lemma \ref{prop:bounds} presents bounds on $e_1$, $e_2$, and $\dot{e}_2$ in healthy and attacked conditions. These bounds are used in the detection logic design and detectability analysis in section \ref{sec:detection}.

\begin{lemma}\label{prop:bounds}
Define $\underline{e}_1\leq e_1 \leq \bar{e}_1$, $\max(\lvert\underline{e}_1\rvert,\lvert \bar{e}_1 \rvert)=\tilde{e}_1$, $\lvert e_2\rvert\leq\tilde{e}_2$,
and $\underline{\dot{e}}_2\leq\lvert \dot{e}_2\rvert\leq\bar{\dot{e}}_2$. Furthermore, denote bounds in healthy conditions, when $\Delta u_{i-1}$=0, with a superscript $0$.

If, elementwise, $\text{diag}(M)>\lvert A_{21}\rvert \tilde{e}_1+\lvert A_{22}\rvert\bar{\zeta}_i+\lvert E_2\rvert\bar{\eta}+\lvert F_2\rvert\bar{\Delta}$, then the following conditions hold
\begin{enumerate}
    \item $\bar{e}_1=\bar{e}_1^0-r_\Delta(\Delta u_{i-1})$
    \item $\underline{e}_1=\underline{e}_1^0-r_\Delta(\Delta u_{i-1})$
    \item $\bar{e}_1^0=e^{A_{11}t}e_1(0) -A_{11}^{-1} (I-e^{A_{11}t})(\lvert A_{12}\rvert\bar{\zeta}_i+\lvert E_1\rvert\bar{\eta})$
    \item $\underline{e}_1^0=e^{A_{11}t}e_1(0) +A_{11}^{-1} (I-e^{A_{11}t})(\lvert A_{12}\rvert\bar{\zeta}_i+\lvert E_1\rvert\bar{\eta})$
    \item $\tilde{e}_2=\bar{\zeta}_i$
    \item $\bar{\dot{e}}^0_2= \lvert A_{21}\rvert \bar{e}^0_1+\lvert A_{22}^{-s}\rvert\bar{\zeta}_i+\lvert E_2\rvert\bar{\eta}+\lvert A_{22}^s \rvert \bar{e}_2+ M$
    \item $\underline{\dot{e}}^0_2=\lvert A_{21}\rvert \underline{e}^0_1-\lvert A_{22}^{-s}\rvert\bar{\zeta}_i-\lvert E_2\rvert\bar{\eta}-\lvert A_{22}^s \rvert \bar{e}_2+ M$
    \item $\bar{\dot{e}}_2=\bar{\dot{e}}^0_2+\text{sgn}(e_y)r(\Delta u_{i-1})$
    \item $\underline{\dot{e}}_2=\underline{\dot{e}}^0_2+\text{sgn}(e_y)r(\Delta u_{i-1})$
    \item $\text{sgn}(\dot{e}_2)=-\text{sgn}(e_y)$
\end{enumerate}
where $r_\Delta(\Delta u_{i-1})=\int_{0}^{t} F_1 \Delta u_{i-1}(s) e^{A_{11}(t-s)}\mbox{d}s$ and $r(\Delta u_{i-1})=A_{21}r_\Delta(\Delta u_{i-1})+F_2\Delta u_{i-1}$.
\end{lemma}
\begin{proof}
Statements 3)-7) and 10) are proven \cite{Keijzer2019}. Statements 1), 2), 8) and 9) are proven in the appendix.
\end{proof}
\section{Detection Logic Design}
\label{sec:detection}
\noindent In this section, the novel cyber-attack detection method will be described. This method directly analyses the behaviour of observer error $e_2$, and uses this to detect cyber-attacks. For comparison, the EOI based detection method presented in previous work \cite{Keijzer2019} is presented in subsection \ref{subsec:EOI_detect}. The detection performance of the two methods in a collaborative IC scenario will be compared in section \ref{sec:sim}.

The novel proposed detection logic uses thresholds on the observer error $e_2$ based on the bounds in Lemma \ref{prop:bounds}, and the relation $e_y=e_2-\zeta_i$. The resulting thresholds, $\underline{e}_2\leq e_2 \leq \bar{e}_2$, will be used for cyber-attack detection. Preferably one would directly monitor this condition, and detect a cyber-attack when it is violated. However, as $e_2$ is not known to the observer, this is not possible. Alternatively, the condition
\begin{equation}\label{eq:detect_cond}
    \underline{e}_2> \bar{e}_2
\end{equation} 
can be monitored. Satisfying this condition implies violation of $\underline{e}_2\leq e_2 \leq \bar{e}_2$, and can thus serve as detection condition.

\subsection{Design of Error-bounds}
\noindent At all time, $\lvert e_2\rvert < \bar{\zeta}_i$ (Lemma ~\ref{prop:bounds}.5), and when a new measurement arrives to the observer $e_y-\bar{\zeta}_i\leq e_2\leq e_y+\bar{\zeta}_i$. Denote the sequence of measurement times as $\{t_m\}$, which do not need to be equidistant. Then, for any time $t_m$

\begin{equation}\label{eq:error_bounds1}
    \begin{aligned}
\bar{e}_2(t_m) =& \min(e_y(t_m)+\bar{\zeta}_i,\bar{\zeta}_i)\\
\underline{e}_2(t_m) =& \max(e_y(t_m)-\bar{\zeta}_i,-\bar{\zeta}_i)\\
\end{aligned}
\end{equation}
Furthermore, bounds on $\dot{e}_2^0$ are known from Lemma~\ref{prop:bounds}. With these, $e_2$ can be bound during each period $[t_{m-1}~t_{m}]$ as
\begin{equation}\label{eq:error_bounds2}
\begin{aligned}
\text{If   } &e_y(t_{m-1}) <0 \\
&\bar{e}_2(t) = \int_{t_{m-1}}^{t}\bar{\dot{e}}_2^0(T)\text{d}T; \underline{e}_2(t) = \int_{t_{m-1}}^{t}\underline{\dot{e}}_2^0(T)\text{d}T\,.\\
\text{If   } &e_y(t_{m-1}) >0 \\
&\bar{e}_2(t) =- \int_{t_{m-1}}^{t}\underline{\dot{e}}_2^0(T)\text{d}T
; \underline{e}_2(t) =- \int_{t_{m-1}}^{t}\bar{\dot{e}}_2^0(T)\text{d}T\,.
\end{aligned}
\end{equation}
The above bounds require further inspection. At first sight they seem to depend only on the modeled healthy system behaviour through $\bar{\dot{e}}_2^0$ and $\underline{\dot{e}}_2^0$. However, the bounds also depend on the real behaviour through $e_y$. The integration duration is dictated by the sign of $e_y$.
The two bounds in equations \eqref{eq:error_bounds1} and \eqref{eq:error_bounds2} can be combined for $m\geq1$ as
\begin{equation*}
\begin{aligned}
\text{If   }e_y(t_{m-1})&<0\\[-0.2cm]
\bar{e}_2(t_m) =& \min(\int_{t_{m-1}}^{t_m}\bar{\dot{e}}_2^0(T)\text{d}T,e_y(t_m)+\bar{\zeta}_i,\bar{\zeta}_i)\\
\underline{e}_2(t_m) =& \max(\int_{t_{m-1}}^{t_m}\underline{\dot{e}}_2^0(T)\text{d}T,e_y(t_m)-\bar{\zeta}_i,-\bar{\zeta}_i)\\[-0.1cm]
\end{aligned}
\end{equation*}
\begin{equation*}
\begin{aligned}
\text{If   } e_y(t_{m-1})&>0\\[-0.2cm]
\bar{e}_2(t_m) =& \min(-\int_{t_{m-1}}^{t_m}\underline{\dot{e}}_2^0(T)\text{d}T,e_y(t_m)+\bar{\zeta}_i,\bar{\zeta}_i)\\
\underline{e}_2(t_m) =& \max(-\int_{t_{m-1}}^{t_m}\bar{\dot{e}}_2^0(T)\text{d}T,e_y(t_m)-\bar{\zeta}_i,-\bar{\zeta}_i)\\
\end{aligned}
\end{equation*}
Equation \eqref{eq:error_bounds1} can be used to obtain $\bar{e}_2(t_0)$ and $\underline{e}_2(t_0)$. Based on these bounds, the detection criterion \eqref{eq:detect_cond} can be monitored at every measurement time $t_m$.

\subsection{Detectability Analysis}
\noindent In this section, conditions are presented for which the proposed novel detection method can detect an attack. Furthermore, it is proven that in healthy conditions, the approach will never cause a detection.

First, introduce an assumption which is required in the presented proofs. This assumption is a relaxation of the matching condition commonly used in SMO literature.\cite{Keijzer2019}
\begin{assumption}\label{ass:rank}
$(F_2-A_{21}A_{11}^\dagger F_1)$ is full column rank.
\end{assumption}
First, it will be proven that no detection occurs in healthy conditions.
\begin{theorem}
Consider system \eqref{eq:SMO_general_model}, observer \eqref{eq:SMO} and detection criterion \eqref{eq:detect_cond}. In healthy conditions, i.e. if $\Delta u_{i-1} = 0 ~ \forall t$, the detection criterion will never be satisfied.
\end{theorem}
\begin{proof}
Define the sequence $\{t_{s_i}\}$ as the times where $e_y$ changes sign, $\dot{e}_2^+$ as the average $\lvert\dot{e}_2\rvert$ while $e_y>0$, and $\dot{e}_2^-$ as the average $\lvert\dot{e}_2\rvert$ while $e_y<0$. Furthermore, without loss of generality, assume $e_y$ becomes positive at every $t_{s_{2i}}$ allowing to write $t^+_i=t_{s_{2i+1}}-t_{s_{2i}}$ and $t^-_i=t_{s_{2i+2}}-t_{s_{2i+1}}$.

Then, denote for the true dynamics of $e_2$ as
\begin{equation}\label{eq:e2_simple}
    \begin{aligned}
        e_2(t_{s_{2i+2}})=e_2(t_{s_{2i}})+c_i\,,\\
        e_2(t_{s_{2i+2N}})=e_2(t_{s_{2i}})+\sum_{j=0}^N c_{i+j} ~ \forall N\in\mathbb{Z}\,,
    \end{aligned}
\end{equation}
where $c_i=t^-_i\dot{e}_2^--t^+_i\dot{e}_2^+$. Now $c_i$ can be bounded, using the bounds on $e_2$ from lemma~\ref{prop:bounds} and $e_y=e_2-\zeta_i$, as
\begin{equation*}
\begin{aligned}
    -e_2(t_{s_{2i}})&+\max(e_y(t_{s_{2i}})-\bar{\zeta}_{i},-\bar{\zeta}_{i}) \leq \sum_{j=i}^{N} c_i \\&\leq -e_2(t_{s_{2i}})+\min(e_y(t_{s_{2i}})+\bar{\zeta}_{i},+\bar{\zeta}_{i}) ~ \forall N\in\mathbb{Z}\,.
\end{aligned}
\end{equation*} 
Furthermore, from equation \eqref{eq:e2_simple}, it can be derived that $\frac{t^+_i}{t^-_i}=\frac{\dot{e}_2^-}{\dot{e}_2^+}+\frac{c_i}{t^-_i\dot{e}_2^+}$. With this, $\bar{e}_2$ in equation \eqref{eq:error_bounds2} can be rewritten as \begin{equation*}
    \bar{e}_2(t_{s_{2i+2}})=\bar{e}_2(t_{s_{2i}})+\frac{t_i^-}{\dot{e}_2^+}(\bar{\dot{e}}_2^0\dot{e}_2^--\underline{\dot{e}}_2^0\dot{e}_2^+)+\frac{\bar{\dot{e}}_2^0}{\dot{e}_2^+}c_{i}\,.
\end{equation*}
which can be extended for $\bar{e}_2(t_{s_{2i+2N}})$ as
\begin{equation}\label{eq:simplified_int_bound_full}
\resizebox{0.89\hsize}{!}{\hspace{-0.3cm}$
    \bar{e}_2(t_{s_{2i+2N}})=\bar{e}_2(t_{s_{2i}})+
    \sum_{j=0}^{N-1}\left(\frac{t_{i+j}^-}{\dot{e}_2^+}(\bar{\dot{e}}_2^0\dot{e}_2^--\underline{\dot{e}}_2^0\dot{e}_2^+)+\frac{\bar{\dot{e}}_2^0}{\dot{e}_2^+}c_{i+j}\right),$}
\end{equation}
for any $N\in\mathbb{Z}$. Similarly for $\underline{e}_2(t_{s_{2i+2N}})$ we can derive
\begin{equation}\label{eq:simplified_int_bound_full-}
\resizebox{0.89\hsize}{!}{\hspace{-0.3cm}$
    \underline{e}_2(t_{s_{2i+2N}})=\underline{e}_2(t_{s_{2i}})+
    \sum_{j=0}^{N-1}\left(\frac{t_{i+j}^-}{\dot{e}_2^+}(\underline{\dot{e}}_2^0\dot{e}_2^--\bar{\dot{e}}_2^0\dot{e}_2^+)+\frac{\underline{\dot{e}}_2^0}{\dot{e}_2^+}c_{i+j}\right).$}
\end{equation}
It can be seen that in healthy conditions, when $\underline{\dot{e}}_2^0\leq\dot{e}_2^-\leq\bar{\dot{e}}_2^0$ and $\underline{\dot{e}}_2^0\leq\dot{e}_2^+\leq\bar{\dot{e}}_2^0$,
\begin{equation}\label{eq:e2changebound}
\begin{aligned}
    \bar{e}_2(t_{s_{2i+2N}})-\bar{e}_2(t_{s_{2i}})&\geq\sum_{j=0}^{N-1}\frac{\bar{\dot{e}}_2^0}{\dot{e}_2^+}c_{i+j}\geq\sum_{j=0}^{N-1}c_{i+j}\\
    \underline{e}_2(t_{s_{2i+2N}})-\underline{e}_2(t_{s_{2i}})&\leq\sum_{j=0}^{N-1}\frac{\underline{\dot{e}}_2^0}{\dot{e}_2^+}c_{i+j}\leq\sum_{j=0}^{N-1}c_{i+j}\\
\end{aligned}
\end{equation}
By subtracting these inequalities it can be found that $\bar{e}_2(t_{s_{2i+2N}})-\underline{e}_2(t_{s_{2i+2N}})\geq \bar{e}_2(t_{s_{2i}})-\underline{e}_2(t_{s_{2i}})$, i.e. considering the behaviour in equation \eqref{eq:error_bounds2}, the difference between $\bar{e}_2$ and $\underline{e}_2$ is non-decreasing. This only leaves to prove that no detection occurs if the bounds are affected by equation \eqref{eq:error_bounds1}.

If both bounds are affected by equation \eqref{eq:error_bounds1}, $\bar{e}_2-\underline{e}_2=2\bar{\zeta}_i-\lvert e_y \rvert\geq 0$. If only the lower bound is affected, use equation \eqref{eq:e2changebound} to derive
\begin{equation*}
\begin{aligned}
    \bar{e}_2(t_{s_{2i+2N}})&\geq\bar{e}_2(t_{s_{2i}})-e_2(t_{s_i})+\max(e_y(t_{s_{2i}})-\bar{\zeta}_{i},-\bar{\zeta}_{i})\\&\geq\max(e_y(t_{s_{2i}})-\bar{\zeta}_{i},-\bar{\zeta}_{i})\geq \underline{e}_2(t_{s_{2i+2N}})
\end{aligned}
\end{equation*}
This proves the theorem.
\end{proof}
Then two lemmas are introduced to support the proof of theorem \ref{thm:detect}, where sufficient conditions for attack detection are presented.
\begin{lemma}\label{lem:r_attack}
consider $r(\Delta u_{i-1})$ as defined in Lemma \ref{prop:bounds}.
Then the following statements can be proven
\begin{enumerate}
\item $r(\Delta u_{i-1})=0~\forall t$ if $\Delta u_{i-1}=0 ~\forall t$, i.e. healthy conditions.
\item There always exists $\gamma>0$ such that within finite time $\lvert r(\Delta u_{i-1})-(F_2-A_{21}A_{11}^\dagger F_1)\Delta u_{i-1}\rvert\leq\gamma\,.$
\end{enumerate}
\end{lemma}
\begin{proof}
By substituting $\Delta u_{i-1}=0 ~\forall t$ in the function for $r(\Delta u_{i-1})$, it can directly be seen that $r(\Delta u_{i-1})=0~\forall t$. This proves statement a). For a constant $\Delta u_{i-1}$, $$\lim_{t\to\infty} r(\Delta u_{i-1})=(F_2-A_{21}A_{11}^\dagger F_1)\Delta u_{i-1}\,.$$ As $r(\Delta u_{i-1})$ is a smooth function, this means that within finite time $\lvert r(\Delta u_{i-1})- (F_2-A_{21}A_{11}^\dagger F_1)\Delta u_{i-1}\rvert<\gamma$ 
\end{proof}
\begin{lemma}\label{lem:detect}
Consider the behaviours of $\bar{e}_2$ and $\underline{e}_2$ from equation \eqref{eq:simplified_int_bound_full} and \eqref{eq:simplified_int_bound_full-}. Assume there exist $\epsilon^+>0$ and $\epsilon^->0$ such that 
\begin{itemize}
    \item $\bar{\dot{e}}_2^0<\dot{e}_2^+-\epsilon^-$ and $\underline{\dot{e}}_2^0>\dot{e}_2^-+\epsilon^-$ for the period $\left[t_{s_{2i}}~ t_{s_{2i+2N}}\right]$.
    \item OR $\underline{\dot{e}}_2^0>\dot{e}_2^++\epsilon^+$ and $\bar{\dot{e}}_2^0<\dot{e}_2^--\epsilon^+$ for the period $\left[t_{s_{2i}}~ t_{s_{2i+2N}}\right]$.
\end{itemize}
Then, there exists an $\epsilon$ such that $\bar{e}_2(t_{s_{2i+2N}})<\underline{e}_2(t_{s_{2i+2N}})$, if $N>\frac{4\bar{\zeta}_i}{\phi\epsilon}$, where $\phi\leq \frac{t^-_{i+j}}{\dot{e}_2^+} ~\forall i,j$.
\end{lemma}
\begin{proof}
First, use $\bar{\dot{e}}_2^0<\dot{e}_2^+-\epsilon^-$ and $\underline{\dot{e}}_2^0>\dot{e}_2^-+\epsilon^-$ to derive
\begin{equation*}
    \bar{\dot{e}}_2^0\dot{e}_2^--\underline{\dot{e}}_2^0\dot{e}_2^+<-(\dot{e}_2^-+\dot{e}_2^+)\epsilon^-<-\epsilon\,.
\end{equation*}
Then substitute $\bar{\dot{e}}_2^0\dot{e}_2^--\underline{\dot{e}}_2^0\dot{e}_2^+<-\epsilon$ and $\phi\leq \frac{t_{i+1}^-}{\dot{e}_2^+} ~\forall i,j$ in equation \eqref{eq:simplified_int_bound_full} giving
\begin{equation*}
    \bar{e}_2(t_{s_{2i+2N}})-\bar{e}_2(t_{s_{2i}})<
    -N\epsilon\phi+\sum_{j=0}^{N-1}c_{i+j}\,.
\end{equation*}
Using the bound on $c_{i+j}$ gives
\begin{equation*}
\begin{aligned}
    &\bar{e}_2(t_{s_{2i+2N}})<\bar{e}_2(t_{s_{2i}})
    -N\epsilon\phi-e_2(t_{s_{2i}})+\min(e_y(t_{s_{2i}})+\bar{\zeta}_i,\bar{\zeta}_i)\\
    &<-N\epsilon\phi+2\min(e_y(t_{s_{2i}})+\bar{\zeta}_i,\bar{\zeta}_i)-\max(e_y(t_{s_{2i}})-\bar{\zeta}_i,-\bar{\zeta}_i)
\end{aligned}
\end{equation*}
Meanwhile, always $\underbar{e}_2(t_{s_{2i+2N}})>\max(e_y(t_{s_{2i}})-\bar{\zeta}_i,-\bar{\zeta}_i)$, which with some simplification leads to
\begin{equation*}
\begin{aligned}
    \bar{e}_2(t_{s_{2i+2N}})-\underline{e}_2(t_{s_{2i+2N}})<-N\epsilon\phi+4\bar{\zeta}_i
\end{aligned}
\end{equation*}
So, $\bar{e}_2(t_{s_{2i+2N}})<\underline{e}_2(t_{s_{2i+2N}})$ if $N>\frac{4\bar{\zeta}_i}{\epsilon\phi}$.

The same result can be obtained by using $\underline{\dot{e}}_2^0>\dot{e}_2^++\epsilon^+$ and $\bar{\dot{e}}_2^0<\dot{e}_2^--\epsilon^+$ to obtain $\underline{\dot{e}}_2^0\dot{e}_2^--\bar{\dot{e}}_2^0\dot{e}_2^+>\epsilon$ and substituting in equation \eqref{eq:simplified_int_bound_full-}.
\end{proof}

\begin{theorem}\label{thm:detect}
Consider system \eqref{eq:SMO_general_model}, with observer \eqref{eq:SMO} and detection criterion \eqref{eq:detect_cond}. There exist a $\delta$ and $\tau$ such that the detection condition \eqref{eq:detect_cond} will be satisfied if $\lvert r(\Delta u_{i-1})\rvert\geq \delta$ for at least a duration $\tau$. Furthermore, if assumption \ref{ass:rank} holds, there always exists a $\Delta u_{i-1}$ such that $\lvert r(\Delta u_{i-1}) \rvert\geq \delta$.
\end{theorem}
\begin{proof}
In Lemma \ref{lem:detect} conditions on $\dot{e}_2^+$ and $\dot{e}_2^-$ are presented such that detection occurs within a duration $\tau=t_{s_{2i+2N}}-t_{s_{2i}}$. Here $N$ is defined in Lemma \ref{lem:detect}. In this proof it thus remains to be shown that there exists a $\delta$ such that the conditions on $\dot{e}_2^+$ and $\dot{e}_2^-$ from Lemma \ref{lem:detect} hold for any attack $r(\Delta u_{i-1})\geq \delta$.

From equation \eqref{eq:e2dot_attacked}, use $\dot{e}_2^+=\dot{e}_2^0+r(\Delta u_{i-1})$ and $\dot{e}_2^-=\dot{e}_2^0-r(\Delta u_{i-1})$. If $r(\Delta u_{i-1})>\bar{\dot{e}}_2^0-\underline{\dot{e}}_2^0+\epsilon^-$, then $\dot{e}_2^+>\dot{e}_2^0+\bar{\dot{e}}_2^0-\underline{\dot{e}}_2^0+\epsilon^-\geq\bar{\dot{e}}_2^0+\epsilon^-$ and $\dot{e}_2^-<\dot{e}_2^0-\bar{\dot{e}}_2^0+\underline{\dot{e}}_2^0-\epsilon^-\leq \underline{\dot{e}}_2^0-\epsilon^-$. This is equivalent to the first condition in Lemma \ref{lem:detect}. The second condition holds if $r(\Delta u_{i-1})<\underline{\dot{e}}_2^0-\bar{\dot{e}}_2^0-\epsilon^+$ and can be proven similarly.

Furthermore, in Lemma \ref{lem:r_attack} it is shown that there exists a $\gamma>0$ such that $\lvert r(\Delta u_{i-1})-(F_2-A_{21}A_{11}^\dagger F_1)\Delta u_{i-1}\rvert\leq\gamma$ within finite time. Therefore, if assumption \ref{ass:rank} holds, there always exists a $\Delta u_{i-1}$ to obtain $\lvert r(\Delta u_{i-1})\rvert \geq\delta$.
\end{proof}

\subsection{Equivalent Output Injection based detection}\label{subsec:EOI_detect}
\noindent In this subsection the equivalent output injection (EOI) based detection, as previously introduced in \cite{Keijzer2019}, is presented for comparison with the novel detection method. First the EOI is defined as
\begin{equation*}
    \dot{\nu}_\text{fil} = K(\nu-\nu_\text{fil})\,,
\end{equation*}
where $\nu_\text{fil}$ is the EOI, and $K\succ0$ is a diagonal gain matrix.

The EOI was originally introduced to estimate the cyber-attack. In \cite{Keijzer2019} the following was proven.
\begin{proposition}
Consider noiseless system \eqref{eq:SMO_general_model}, where $\zeta_i=0$, and the SMO \eqref{eq:SMO}. If $\text{diag}(M)>\lvert A_{21}\rvert \tilde{e}_1+\lvert A_{22}\rvert\bar{\zeta}_i+\lvert E_2\rvert\bar{\eta}+\lvert F_2\rvert\bar{\Delta}$, then $e_2\to 0$ and $\dot{e}_2\to 0$. Furthermore, assuming a constant cyber-attack, the cyber-attack estimate $$\hat{\Delta} u_{i-1}=(F_2-A_{21}A_{11}^{-1}F_1)^\dagger \nu_\text{fil}$$ has an accuracy of $$\lvert\Delta u_{i-1}-\hat{\Delta} u_{i-1}\rvert \leq\lvert(F_2-A_{21}A_{11}^{-1}F_1)^\dagger( A_{21}A_{11}^{-1}\lvert E_1\rvert+\lvert E_2\rvert) \bar{\eta}\rvert$$\hfill $\blacksquare$
\end{proposition}
The EOI can also be used for cyber-attack detection. Based on the bounds presented in Lemma \ref{prop:bounds}, a threshold for EOI-based cyber-attack detection is introduced in \cite{Keijzer2019}, globally bounding the healthy EOI behaviour. In the threshold, each element $_{(i)}$ is defined as
\begin{equation*}
    \bar{\nu}_{\text{fil},(i)} = e^{-k\bar{t}^{0,*}_{(i)}}\bar{U}_{(i)} + (1-e^{-k\bar{t}^{0,*}_{(i)}})m\,,
\end{equation*}
where $k=K_{(i,i)}$, $m=M_{(i,i)}$, $\bar{t}^{0,*}=\lim_{t\to\infty}\frac{2\bar{e}_{2}}{\underline{\dot{e}}_{2}^0}$, and $\bar{U} =\lim_{t\to\infty} \lvert A_{12}\rvert \bar{e}_1^{0}+\lvert A_{22}^{-s}\rvert\bar{\zeta}_i+\lvert E_2\rvert\bar{\eta}+\lvert A_{22}^s \rvert \bar{e}_2$. A lower threshold $\underline{\nu}_\text{fil}=-\bar{\nu}_\text{fil}$ is derived similarly. A cyber-attack is detected if the condition $\underline{\nu}_\text{fil}\leq\nu_\text{fil}\leq\bar{\nu}_\text{fil}$ is violated.
\section{Simulation of Intersection Control}
\label{sec:sim}
\noindent A simulation is performed with 2 cars approaching an intersection. The car closest to the intersection will be referred to as the leader car, for which the input sequence is pre-defined. The car furthest from the intersection is the follower car, and is controlled using the control law from \cite{Ploeg2011a} shown below. The detection algorithm works regardless of the control law.
\begin{equation*}
    \dot{u}_i = -\frac{1}{h}(u_i+k_p\epsilon_1+k_d\dot{\epsilon}_1-u_{i-1})\,.
\end{equation*}
Here $\epsilon_1=d_i-d_{i,r}$, $d_i=p_i-p_{i-1}-L_{i-1}$ is the relative distance of the cars to the intersection, and $d_{i,r} = r+hv_i$ is the desired relative distance. The time headway $h$, standstill distance $r$, and control gains $k_p$ and $k_d$, are defined in Table~\ref{tab:sim_param}.
\begin{table}[h]
    \centering
    \caption{Parameters used in simulation}
    \label{tab:sim_param}
    \begin{tabular}{c|c||c|c}
        Parameter & Value & Parameter & Value \\\hline
        $p_0(0)$ & $-40\,[m]$ & $p_1(0)$&$-50\,[m]$\\
        $v_0(0)$ & $8\,\left[\frac{m}{s}\right]$ & $v_1(0)$ & $10\,\left[\frac{m}{s}\right]$\\
        $a_0(0)$ & $0\,\left[\frac{m}{s^2}\right]$& $a_1(0)$&$0\,\left[\frac{m}{s^2}\right]$\\\hline
        $\tau_0$ &$0.11\,[s]$ & $\tau_1$& $0.1\,[s]$\\
        $L_1$ & $4\,[m]$ &$r_{\tau}$& $0.9\,[-]$\\\hline
        $h$ & $0.7\,[s]$ & $r$& $1.5\,[m]$\\
        $k_p$ & $0.2\,[s^{-2}]$ &$k_d$ & $0.7\,[s^{-1}]$\\\hline
        $\bar{\Delta}$ & $10\,\left[\frac{m}{s^2}\right]$& $\bar{\eta}$& $1\,\left[\frac{m}{s^2}\right]$\\
        $\bar{\zeta}_1$ & \multicolumn{3}{c}{$[0.15~0.3~0.03~0.15]^\top\,\left[m~\frac{m}{s}~\frac{m}{s}~\frac{m}{s^2}\right]$}\\\hline
        $K$ & $I_4\,[s^{-1}]$& $A_{22}^s$& $-0.1\cdot I_4\,[s^{-1}]$\\
        $M$ & \multicolumn{3}{c}{$\text{diag}([0.5~11.5~0.2~2.0])\,\left[m~\frac{m}{s}~\frac{m}{s}~\frac{m}{s^2}\right]$}\\
    \end{tabular}
\end{table}

Based on \cite{Abdulla2016}, IC is initiated when the cars are within $50 [m]$ from the intersection. Furthermore, the intersection is approached at $8 [m/s]\approx 30 [km/h]$, which is a common standard speed in urban areas.
In figure \ref{fig:SimScenario}, the input of the leader car, and the cyber-attack considered are show. It can be seen that the lead vehicle drives at a constant speed, and at $t = 0.5 [s]$ a step attack is performed on the communication.
In figure \ref{fig:SimPos} it is shown that this attack causes the cars to gradually drive closer together, eventually leading to a crash at the intersection. The crash occurs when the follower vehicle enters the intersection, i.e. $p_f=-4~[m]$. At this point the lead vehicle has not yet left the intersection, i.e. $-4~[m]<p_l<0~[m]$, resulting in a crash.
\begin{figure}[h]
	\centering
	\includegraphics[width=0.8\columnwidth, trim={0cm 0cm 7cm 20.5cm},clip]{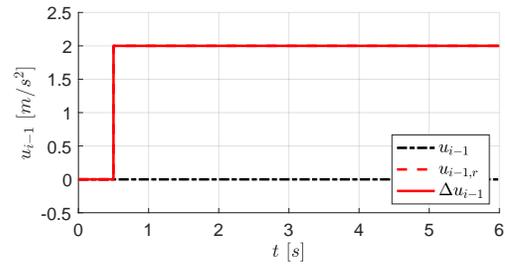}
	\caption{Input of leader car and considered attack on the communication.}
	\label{fig:SimScenario}
\end{figure}
\begin{figure}[h]
	\centering
	\includegraphics[width=0.8\columnwidth, trim={0cm 0cm 7cm 20.6cm},clip]{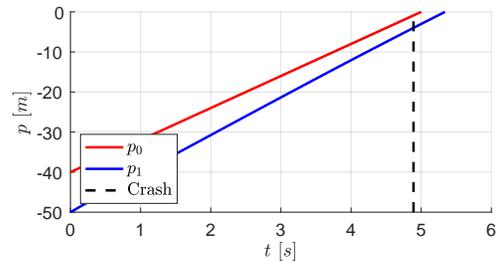}
	\caption{Effect of the attack on vehicle positions, leading up to the crash.}
	\label{fig:SimPos}
\end{figure}
\begin{figure}[h]
	\centering
	\includegraphics[width=0.8\columnwidth, trim={0cm 0cm 7cm 18.3cm},clip]{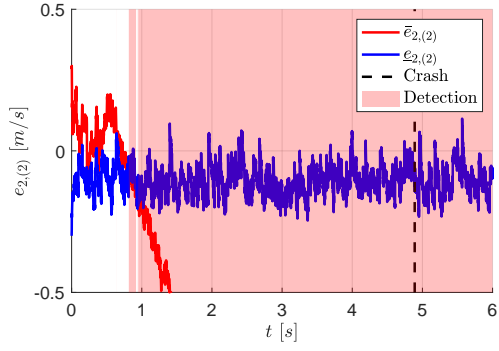}
	\caption{Error bounds and detection performance of novel detection logic.}
	\label{fig:Sim_ebound}
\end{figure}
\begin{figure}[h]
	\centering
	\includegraphics[width=0.8\columnwidth, trim={0cm 0cm 7cm 18.4cm},clip]{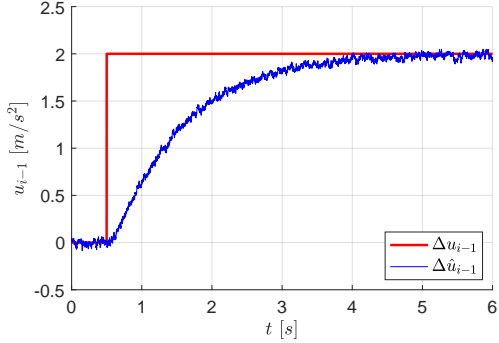}
	\caption{Attack estimation by EOI based estimation method.}
	\label{fig:Sim_nuest}
\end{figure}
\begin{figure}[h]
	\centering
	\includegraphics[width=0.8\columnwidth, trim={0cm 0cm 7cm 18.3cm},clip]{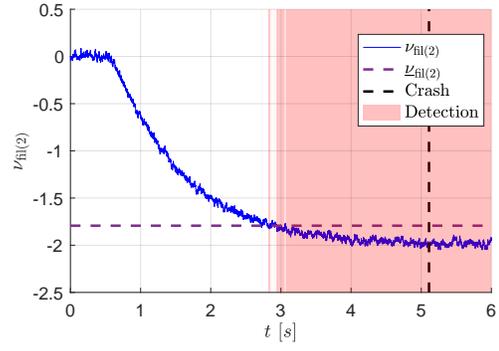}
	\caption{Attack detection by EOI based detection method.}
	\label{fig:Sim_nudetect}
\end{figure}
Figure \ref{fig:Sim_ebound} shows the detection results obtained with the novel detection logic presented in Section \ref{sec:detection}. Detection first occurs for a very short period at $t=0.64~[s]$. This is not visible in the figure. More consistent detection occurs at $t=0.82~[s]$. This consistent detection occurs well before the crash occurs at $t=4.8~[s]$.

In Figure \ref{fig:Sim_nudetect}, it is shown that the EOI based detection method from previous work also detects the attack. However, the detection only occurs at $t=2.73~[s]$. This is significantly slower than the detection with the novel detection method.

The capability of the EOI based method to also estimate the attack is illustrated in figure \ref{fig:Sim_nuest}. This estimate of the attack can be very useful in designing an effective control strategy to deal with the attack.

As both detection methods depend on the same observer, it is feasible to implement both methods concurrently. In this way the best properties of both methods can be combined.
\section{Conclusions}
\label{sec:conclusion}
\noindent Intersections, in the way they are currently designed for human drivers, are a major cause of accidents as well as traffic delays. Therefore, automated control of vehicles in intersections offers great potential for improvement. However, as Intersection Control systems do rely on wireless V2X networks for traffic coordination, security of such communication channel is paramount for safety. In V2X networks a first line of preventive security measures are already in place at the protocol level to make cyber-attacks more difficult. Still the possibility of an inside attacker, or one capable of overcoming preventive security measures cannot be ruled out. To protect against such a scenario, a second line of defenses based on a cyber-attack detection method is proposed in this paper.
In particular, a novel detector based on a Sliding Mode Observer and a corresponding set of thresholds was designed. With respect to previous results, the novel detection approach is shown to be faster and more sensitive, as the filtering of the observer Equivalent Output Injection is avoided. Theoretical results certifying the robustness and detectability of the proposed approach were provided, as well as a simulation study.

In the future, a two-stage approach based on a fault detector and a fault identification scheme may be envisaged, thus paving the way for fully autonomous accommodation of faults and cyber-attacks in Intersection Control systems based on cooperative autonomous vehicles. Furthermore, adaptations to the detection method are envisioned for which boundedness of the attack is no longer required.

\appendix
\noindent From the first row of equation \eqref{eq:SMO_err_dyn}, using lemma 1.1.1 in \cite{Lakshmikantham2015-il} we obtain 
$e_1(t) = e^{A_{11}(t)}e_1(0) -r_\zeta(\zeta_i) -r_\eta(\eta)-r_\Delta(\Delta u_{i-1})$
, where $r_\zeta(\zeta_i) = \int_{0}^{t} A_{21} \zeta_i(s) e^{A_{11}(t-s)}\mbox{d}s$ and $r_\eta(\eta) = \int_{0}^{t} E_1 \eta(s) e^{A_{11}(t-s)}\mbox{d}s$.

From this it can be concluded that $e_1(t)=e_1^0(t)-r_\Delta(\Delta u_{i-1})$, and therefore
$    \bar{e}_1(t)=\bar{e}_1^0(t)-r_\Delta(\Delta u_{i-1})$ and  $\underline{e}_1(t)=\underline{e}_1^0(t)-r_\Delta(\Delta u_{i-1})$
, which proves statements 1) and 2) in Lemma \ref{prop:bounds}.

Furthermore $\dot{e}_2$ from the second row of equation \eqref{eq:SMO_err_dyn} is
$    \dot{e}_2 = A_{21}e_1+A_{22}^s e_2 +A_{22
    }^{-s}\zeta_i-E_2\eta-F_2\Delta u_{i-1} - M\text{sgn}(e_y)$ .
With this we can write, 
\begin{equation}\label{eq:e2dot_attacked}
\resizebox{0.89\hsize}{!}{\hspace{-0.4cm}$
        \dot{e}_2 = \dot{e}_2^0 + A_{21}(e_1-e_1^0)-F_2\Delta u_{i-1} =
        \dot{e}_2^0 -A_{21}r_\Delta(\Delta u_{i-1})-F_2\Delta u_{i-1}.$}
\end{equation}
To prove statements 8) and 9) in Lemma \ref{prop:bounds} we finally observe that
 $   \bar{\dot{e}}_2 = \bar{\dot{e}}_2^0 +\text{sgn}(e_y)r(\Delta u_{i-1})$ and
    $\underline{\dot{e}}_2 = \underline{\dot{e}}_2^0 +\text{sgn}(e_y)r(\Delta u_{i-1})$.

\begin{spacing}{0.85}
\bibliographystyle{IEEEtran}
\bibliography{references}

\begin{thebibliography}{10}
\providecommand{\url}[1]{#1}
\csname url@samestyle\endcsname
\providecommand{\newblock}{\relax}
\providecommand{\bibinfo}[2]{#2}
\providecommand{\BIBentrySTDinterwordspacing}{\spaceskip=0pt\relax}
\providecommand{\BIBentryALTinterwordstretchfactor}{4}
\providecommand{\BIBentryALTinterwordspacing}{\spaceskip=\fontdimen2\font plus
\BIBentryALTinterwordstretchfactor\fontdimen3\font minus
  \fontdimen4\font\relax}
\providecommand{\BIBforeignlanguage}[2]{{%
\expandafter\ifx\csname l@#1\endcsname\relax
\typeout{** WARNING: IEEEtran.bst: No hyphenation pattern has been}%
\typeout{** loaded for the language `#1'. Using the pattern for}%
\typeout{** the default language instead.}%
\else
\language=\csname l@#1\endcsname
\fi
#2}}
\providecommand{\BIBdecl}{\relax}
\BIBdecl

\bibitem{european2016traffic}
{European Commission}, ``Traffic safety basic facts on junctions,'' 2016.

\bibitem{dresner2005multiagent}
K.~Dresner and P.~Stone, ``Multiagent traffic management: An improved
  intersection control mechanism,'' in \emph{AAMAS'05, Procs of}, 2005.

\bibitem{bondavalli_making_2014}
T.~Arts, M.~Dorigatti, and S.~Tonetta, ``\BIBforeignlanguage{en}{Making
  {Implicit} {Safety} {Requirements} {Explicit}: {An} {AUTOSAR} {Safety}
  {Case}},'' in \emph{\BIBforeignlanguage{en}{Computer {Safety}, {Reliability},
  and {Security}}}, A.~Bondavalli and F.~Di~Giandomenico, Eds.\hskip 1em plus
  0.5em minus 0.4em\relax Cham: Springer International Publishing, 2014, vol.
  8666, pp. 81--92.

\bibitem{amoozadeh2015security}
M.~Amoozadeh, A.~Raghuramu, C.-N. Chuah, D.~Ghosal, H.~M. Zhang, J.~Rowe, and
  K.~Levitt, ``Security vulnerabilities of connected vehicle streams and their
  impact on cooperative driving,'' \emph{IEEE Communications Magazine},
  vol.~53, no.~6, pp. 126--132, 2015.

\bibitem{yan2016can}
C.~Yan, W.~Xu, and J.~Liu, ``Can you trust autonomous vehicles: Contactless
  attacks against sensors of self-driving vehicle,'' \emph{DEF CON}, vol.~24,
  no.~8, p. 109, 2016.

\bibitem{engoulou2014vanet}
R.~G. Engoulou, M.~Bella{\"\i}che, S.~Pierre, and A.~Quintero, ``Vanet security
  surveys,'' \emph{Computer Communications}, vol.~44, pp. 1--13, 2014.

\bibitem{van2017analyzing}
R.~van~der Heijden, T.~Lukaseder, and F.~Kargl, ``Analyzing attacks on
  cooperative adaptive cruise control (cacc),'' in \emph{2017 IEEE Vehicular
  Networking Conference (VNC)}.\hskip 1em plus 0.5em minus 0.4em\relax IEEE,
  2017, pp. 45--52.

\bibitem{bissmeyer2012assessment}
N.~Bi{\ss}meyer, S.~Mauthofer, K.~M. Bayarou, and F.~Kargl, ``Assessment of
  node trustworthiness in vanets using data plausibility checks with particle
  filters,'' in \emph{VNC, Procs. of}, 2012.

\bibitem{dietzel2012graph}
S.~Dietzel, J.~Petit, G.~Heijenk, and F.~Kargl, ``Graph-based metrics for
  insider attack detection in vanet multihop data dissemination protocols,''
  \emph{IEEE Trans. on Vehic. Tech.}, vol.~62, no.~4, pp. 1505--1518, 2012.

\bibitem{yan2008providing}
G.~Yan, S.~Olariu, and M.~C. Weigle, ``Providing vanet security through active
  position detection,'' \emph{Computer communications}, vol.~31, no.~12, pp.
  2883--2897, 2008.

\bibitem{pasqualetti2013attack}
F.~Pasqualetti, F.~D{\"o}rfler, and F.~Bullo, ``Attack detection and
  identification in cyber-physical systems,'' \emph{IEEE Trans. on Automatic
  Control}, vol.~58, no.~11, pp. 2715--2729, 2013.

\bibitem{quan2018distributed}
Y.~Quan, W.~Chen, Z.~Wu, and L.~Peng, ``Distributed fault detection and
  isolation for leader--follower multi-agent systems with disturbances using
  observer techniques,'' \emph{Nonlinear Dynamics}, vol.~93, no.~2, pp.
  863--871, 2018.

\bibitem{Jahanshahi_2018aa}
N.~Jahanshahi and R.~M. Ferrari, ``Attack detection and estimation in
  cooperative vehicles platoons: A sliding mode observer approach,'' in
  \emph{NECSYS, Procs. of}, August 2018.

\bibitem{Keijzer2019}
T.~Keijzer and R.~M. Ferrari, ``A sliding mode observer approach for attack
  detection and estimation in autonomous vehicle platoons using event triggered
  communication,'' in \emph{CDC, Procs. of}, 2019.

\bibitem{Keijzer_2021aa}
------, ``Detection of network and sensor cyber-attacks in platoons of
  cooperative autonomous vehicles: a sliding-mode observer approach,'' in
  \emph{ECC, Procs. of}, 2021.

\bibitem{medina2017cooperative}
A.~I.~M. Medina, N.~van~de Wouw, and H.~Nijmeijer, ``Cooperative intersection
  control based on virtual platooning,'' \emph{IEEE Trans. on Intelligent
  Transp. Systems}, vol.~19, no.~6, pp. 1727--1740, 2017.

\bibitem{campos2014cooperative}
G.~R. Campos, P.~Falcone, H.~Wymeersch, R.~Hult, and J.~Sj{\"o}berg,
  ``Cooperative receding horizon conflict resolution at traffic
  intersections,'' in \emph{CDC, Procs. of}, 2014.

\bibitem{kamal2014vehicle}
M.~A.~S. Kamal, J.-i. Imura, T.~Hayakawa, A.~Ohata, and K.~Aihara, ``A
  vehicle-intersection coordination scheme for smooth flows of traffic without
  using traffic lights,'' \emph{IEEE Transactions on Intelligent Transportation
  Systems}, vol.~16, no.~3, pp. 1136--1147, 2014.

\bibitem{katriniok2017distributed}
A.~Katriniok, P.~Kleibaum, and M.~Jo{\v{s}}evski, ``Distributed model
  predictive control for intersection automation using a parallelized
  optimization approach,'' in \emph{IFAC World Congress}, 2017.

\bibitem{kneissl2018feasible}
M.~Kneissl, A.~Molin, H.~Esen, and S.~Hirche, ``A feasible mpc-based
  negotiation algorithm for automated intersection crossing,'' in \emph{2018
  European Control Conference (ECC)}.\hskip 1em plus 0.5em minus 0.4em\relax
  IEEE, 2018, pp. 1282--1288.

\bibitem{Bosch_sensor2021}
\BIBentryALTinterwordspacing
{Robert Bosch GmbH}. (2021) Corner radar sensor. [Online]. Available:
  \url{https://www.bosch-mobility-solutions.com/en/products-and-services/passenger-cars-and-light-commercial-vehicles/driver-assistance-systems/automatic-emergency-braking/corner-radar-sensor/}
\BIBentrySTDinterwordspacing

\bibitem{ElineThesis}
E.~Janse, ``Anomaly detection in intersection control: Sliding mode observer
  based anomaly detection in virtual platooning enabled intersection control,''
  Master's thesis, TU Delft, 2020.

\bibitem{Edwards2000-vs}
C.~Edwards, S.~K. Spurgeon, and R.~J. Patton, ``Sliding mode observers for
  fault detection and isolation,'' \emph{Automatica}, vol.~36, no.~4, pp.
  541--553, Apr. 2000.

\bibitem{Ploeg2011a}
J.~Ploeg, B.~T.~M. Scheepers, E.~van Nunen, N.~van~de Wouw, and H.~Nijmeijer,
  ``{Design and experimental evaluation of cooperative adaptive cruise
  control},'' in \emph{14th International IEEE Conference on Intelligent
  Transportation Systems}, 2011, pp. 260--265.

\bibitem{Abdulla2016}
M.~{Abdulla}, E.~{Steinmetz}, and H.~{Wymeersch}, ``Vehicle-to-vehicle
  communications with urban intersection path loss models,'' in \emph{2016 IEEE
  Globecom Workshops (GC Wkshps)}, 2016, pp. 1--6.

\bibitem{Lakshmikantham2015-il}
V.~Lakshmikantham, S.~Leela, and A.~A. Martynyuk, \emph{Stability Analysis of
  Nonlinear Systems}.\hskip 1em plus 0.5em minus 0.4em\relax Birkh{\"a}user,
  2015.

\end{thebibliography}
\end{spacing}
\end{document}